\documentclass[11pt]{article}
\usepackage[letterpaper,hmargin=26mm,vmargin=1in]{geometry}
\usepackage{graphicx,amssymb,amsmath,amsthm}

\usepackage{algorithm}
\usepackage{algorithmic}

\usepackage{latexsym}
\usepackage{epic}
\usepackage{eepic}

\newtheorem{theorem}{Theorem}[section]

\newtheorem{proposition}[theorem]{Proposition}

\theoremstyle{definition}

\newtheorem{definitions}[theorem]{Definitions}
\newtheorem{example}[theorem]{Example}

\newcommand{\ignore}[1]{}

\begin{document}

\title{False name manipulations in weighted voting games: splitting, merging and annexation}

\author{
\em Haris Aziz
\and \em Mike Paterson
}

\date{} %% date removed

\maketitle

\begin{abstract}
An important aspect of mechanism design in social choice protocols and multiagent systems is to discourage insincere and manipulative behaviour. We examine the computational complexity of false-name manipulation in weighted voting games which are an important class of coalitional voting games. Weighted voting games have received increased interest in the multiagent community due to their compact representation and ability to model coalitional formation scenarios. Bachrach and Elkind in their AAMAS 2008 paper examined divide and conquer false-name manipulation in weighted voting games from the point of view of Shapley-Shubik index. We analyse the corresponding case of the Banzhaf index and check how much the Banzhaf index of a player increases or decreases if it splits up into sub-players. A pseudo-polynomial algorithm to find the optimal split is also provided. Bachrach and Elkind also mentioned manipulation via merging as an open problem. In the paper, we examine the cases where a player annexes other players or merges with them to increase their Banzhaf index or Shapley-Shubik index payoff. We characterize the computational complexity of such manipulations and provide limits to the manipulation. The \emph{annexation non-monotonicity paradox} is also discovered in the case of the Banzhaf index. The results give insight into coalition formation and manipulation.
\end{abstract}

\section{Introduction}
 
\subsection{Motivation}

Weighted voting games (WVGs) are mathematical models which are used to analyze voting bodies in which the voters have different numbers of votes.  In WVGs, each voter is assigned a non-negative weight and makes a vote in favour of or against a bill. The bill is passed if and only if the total weight of those voting in favour of the bill is greater than or equal to some fixed quota. Power indices such as the Banzhaf index measure the ability of a player in a WVG to determine the outcome of the vote. 

WVGs have received increased interest in the artificial intelligence and agents community due to their ability to model various coalition formation  scenarios \cite{Dimensionality-of-Voting-Games,DBLP:conf/aaai/ElkindGGW07}. Such games have also been examined from the point of view of susceptibility to manipulations \cite{edithideas,Zuckerman:2008:AAAI}. WVGs and coalitional voting games are also encountered in threshold logic, reliability theory, neuroscience and logical computing devices \cite{simplegames, 45427, coherentandsimplegames}. WVGs have been applied in various political and economic organizations~\cite{leech_imf, Laurelle, bilbao_2007}. Voting power is used in joint stock companies where each shareholder gets votes in proportion to the ownership of a stock \cite{gambarelli2}. 

Elkind et al. \cite{Dimensionality-of-Voting-Games} note that since WVGs have only two possible outcomes, they do not fall prey to manipulation of the type characterized by Gibbard-Satterthwaite~\cite{RePEc:ecm:emetrp:v:41:y:1973:i:4:p:587-601}. However, there are various ways WVGs can be manipulated and controlled. Splitting of a player into sub-players can be seen as a false-name manipulation by an agent where it splits itself into multiple agents so that the sum of the utilities of the split-up players is more than the utility of the original player~\cite{edithideas}. We examine situations when a player splitting up into smaller players may be advantageous or disadvantageous in the context of WVGs and Banzhaf indices. This gives a better idea of how to devise WVGs in which manipulation can be deterred. This may be done by keeping larger or non-integer weights. Moreover, we also examine the case of players merging to maximize their payoff in a WVG. This was mentioned as an unexplored question in \cite{edithideas}.

\subsection{Outline}

In Section~\ref{Preliminaries-LAMSADE}, some basic definitions concerning simple games, weighted voting games and power indices are provided. Section \ref{AAMAS2009-relatedwork} provides a brief literature survey. In Section~\ref{Splitting-LAMSADE}, the case of players splitting up into sub-players in a WVG to increase their Banzhaf index is analysed. We examine the extent to which the Banzhaf index of a player can increase or decrease if it splits up into sub-players. From a computational perspective, it is \#P-hard~\cite{91928} to compute the payoff in the WVG. A prospective manipulator could still be interested in enabling a beneficial split even if he cannot compute the actual payoff. Moreover, this model is reasonable because the centre is assumed to have much more computational resources than the players. In Section~\ref{AAMAS-COMPLEXITY-SECTION}, we prove that it is NP-hard even to decide whether a split is beneficial or not. In the end a pseudo-polynomial algorithm is proposed which returns `no' if no beneficial split is available and returns the optimal split otherwise. Section \ref{MERGING-SECTION-AAMAS} is about the case of players annexing others or voluntarily merging into blocs to maximize their payoffs. It is shown that it is NP-hard to decide a beneficial merge for both the Banzhaf index and the Shapley-Shubik index and to decide a beneficial annexation for the case of the Banzhaf index.  Limits to manipulation are also provided. The final section presents conclusions and ideas for future work.

\section{Preliminaries}\label{Preliminaries-LAMSADE}

In this section we give definitions and notations of key terms. The set of voters is $N = \{1,...,n\}$. 

\begin{definitions}
A \emph{simple voting game} is a pair $(N,v)$ where the valuation function $v:2^N \rightarrow \{0,1\}$ has the properties that $v(\emptyset)=0$, $v(N)=1$ and $v(S)\leq v(T)$ whenever $S \subseteq T$. A \emph{coalition} $S \subseteq N$ is \emph{winning} if $v(S)=1$ and \emph{losing} if $v(S)=0$. A simple voting game can alternatively be defined as $(N,W)$ where $W$ is the set of winning coalitions.
\end{definitions}

\begin{definitions}
The simple voting game $(N,W)$ where \\
$W=\{X \subseteq N, \sum_{x \in X}{w_x} \geq q \}$ is called a \emph{weighted voting game} (WVG). A WVG is denoted by $[q;w_1,w_2,...,w_n]$ where $w_i \geq 0$ is the voting weight of player $i$. By convention, we take $w_i \geq w_j$ if $i<j$.
\end{definitions}

Usually, $q>\frac{1}{2}{\sum_{1 \leq i \leq n}w_i}$ so that there are no two mutually exclusive winning coalitions at the same time. WVGs with this property are termed \emph{proper}. Proper WVGs are desirable because more than a majority is necessary to force the decision.

%they satisfy the criterion of the majority getting preference. 

\begin{definitions}
A player $i$ is \emph{critical in} a winning coalition $S$ when $S \in W$ and $S \setminus \{i\} \notin W$. For each $i \in N$, we denote the number of coalitions in which $i$ is critical in game $v$ by ${{\eta}_{i}}(v)$. The \emph{Banzhaf index} of player $i$ in WVG $v$ is $\beta_i = \frac{{{\eta}_{i}}(v)}{{\sum}_{i \in N}{{\eta}_{i}}(v)}$. The \emph{probabilistic Banzhaf index}, $\beta_i^{'}$ of player $i$ in game $v$ is equal to  ${{{\eta}_{i}}(v)}/2^{n-1}$.
\end{definitions}

\begin{definitions}
The \emph{Shapley-Shubik value} is the function $\kappa$ that assigns to any simple game $(N,v)$ and any voter $i$ a value $\kappa_{i}
(v)$ where $\kappa_{i} = \sum_{X \subseteq N} (|X|-1)!(n-|X|)!(v(X)- v(X-\{i\}))$. The \emph{Shapley-Shubik index} of $i$ is the function $\phi$ defined by $\phi_i=\frac{{\kappa}_{i}}{n!}$
\end{definitions}

%\subsection{Dictators, dummies and veto powers}

%\begin{definition} 
%\emph{Dictator} is a player who is present in every winning coalition and absent from every losing coalition
%\end{definition}

%The definition implies that the dictator is player $1$ with the biggest weight, $\beta_1=1$ and $\forall i \neq 1 \beta_i=0$. This is because any other player cannot prove critical in any coalition. Any winning coalition has to include the dictator and the opting out of any other player cannot make the coalition losing. For player $1$ to be the dictator, $w_1 \geq q$ and $\sum_{2 \leq i \leq n } w_i<q $. It is obvious that if a dictator exists, he is unique.

%\begin{definition} 
%A player has \emph{veto power} if and only if the player is present in every winning coalition.
%\end{definition}

%It is evident that a dictator has veto power but a player with veto power is not necessarily a dictator. Moreover there can be two or more players with veto power.

\section{Related work}\label{AAMAS2009-relatedwork}

Weighted voting games date back at least to John von Neumann and Oskar Morgenstern who developed their theory in their monumental book \emph{Theory of Games and Economic Behavior}~\cite{1944book}. WVGs and voting power indices have been analyzed extensively in the game theory literature for instance in \cite{Dubey1979,Roth1988}. They have been applied to various economic and political bodies such as the EU Council of Ministers and the IMF \cite{leech_imf}. Power indices such as the Banzhaf index and the Shapley-Shubik index originated in such a setting in order to gauge the decision making ability of players. These indices have now been utilized in different domains such as networks \cite{DBLP:conf/atal/BachrachRP08}. Simple games and weighted voting games are known by different names in other literatures and communities. There is considerable work on similar models in threshold logic \cite{murogabookthresholdlogic}.

As useful and succinct models for coalitional voting games, WVGs have been utilized in multiagent systems. Voting power indices in WVGs have received increased interest in multiagent systems \cite{DBLP:conf/atal/FatimaWJ08,DBLP:conf/atal/BachrachMPRS08, Aziz2007a, Aziz-comsoc2008}. The dimension of a multiple weighted voting game is the minimum number of weighted voting games required to represent it. The dimension of multiple weighted voting games has been examined in \cite{Dimensionality-of-Voting-Games} and \cite{DBLP:journals/eor/DeinekoW06}. Moreover, the complexity of questions related to important cooperative game solutions in WVGs such as the core and nucleolus are considered in \cite{DBLP:conf/aaai/ElkindGGW07}. WVGs have also been examined from the point of view of control and manipulation. Zuckerman et al. \cite{Zuckerman:2008:AAAI} analyse how the centre might control WVGs by changing the quota even if the weights are fixed. The most relevant work is by Elkind et al. \cite{edithideas} where they examine false-name manipulation in WVGs from the point of view of the Shapley-Shubik index. In fact, our paper answers problems posed by Elkind et al. Players forming blocs have been considered by political scientists and economists previously \cite{mainbook}. However, in this paper, a complexity theoretic analysis of bloc forming manipulation has also been undertaken for WVGs. False name manipulations in open anonymous environments have been examined in different domains such as coalitional games~\cite{DBLP:conf/aaai/YokooCSOI05, DBLP:conf/aaai/OhtaIYMCS06, DBLP:conf/atal/OhtaCSIY08} and auctions~\cite{1345049,DBLP:conf/wine/IwasakiKSSY07}. The characteristic function by itself does not give enough
information to analyze false-name manipulations especially if a player splits into sub-players. Therefore Yokoo et al.~\cite{DBLP:conf/aaai/YokooCSOI05} introduced the model where each player has a subset of skills and the characteristic function assigns values to the subset of skills. We notice that false-name manipulations in WVGs can still be analyzed directly without considering more fine-grained representations.

\section{Splitting}\label{Splitting-LAMSADE}

\subsection{Background}

In the real world, WVGs may be dynamic. Players might have an incentive to split up into smaller players or merge into voting blocks. %{\em Imputations} of the players are acceptable distributions of the payoff of the grand coalition. 
Payoffs of players in a coalitional games setting can be based on fairness, i.e., power indices, or they can be based on the notion of stability, which includes many cooperative game theoretic concepts such as core, nucleolus etc. We examine the situation when the Banzhaf indices of agents can be used as payoffs in a cooperative game theoretic situation. Falsenthal and Machover~\cite{RePEc:spr:sochwe:v:19:y:2002:i:2:p:295-312} refer to this notion of voting power as P-power since the motivation of agents is prize-seeking as opposed to influence-seeking. However Banzhaf indices have been considered as possible payments in cooperative settings~\cite{van:August2002:0040-5833:61, edithideas} and they satisfy desirable axioms~\cite{Dubey1979}. Splitting of a player can be seen as a false-name manipulation by an agent, in which it splits itself into multiple agents so that the sum of the utilities of the split-up players is more than the utility of the original player~\cite{edithideas}. 

Splitting is not always beneficial. We give examples where, if we use Banzhaf indices as payoffs of players in a WVG, splitting can be advantageous, neutral or disadvantageous. %This analysis of a player splitting into sub players is different from the problem of independent agents forming blocs since in that case, every agent has different incentives which leads to a two-tier voting where the agents first vote for a decision and then the cumulative weight for that decision is carried in the second tier.

\begin{example}
Splitting can be advantageous, neutral or disadvantageous:
\begin{itemize}
\item \emph{Disadvantageous splitting}. In the WVG $[5;2,2,2]$ each player has a Banzhaf index of $1/3$. If the last player splits up into two players, the new game is $[5;2,2,1,1]$. In that case, the split-up players have a Banzhaf index of $1/8$ each. 
\item \emph{Neutral splitting}. In the WVG $[4;2,2,2]$ each player has a Banzhaf index of $1/3$. If the last player splits up into two players, the new game is $[4;2,2,1,1]$. In that case, the split-up players have a Banzhaf index of $1/6$ each.
\item \emph{Advantageous splitting}. In the WVG $[6;2,2,2]$ each player has a Banzhaf index of $1/3$. If the last player splits up into two players, the new game is $[6;2,2,1,1]$. In that case, the split-up players have a Banzhaf index of $1/4$ each.
\end{itemize}
\end{example}

We analyse the splitting of players in the unanimity WVG.

\begin{proposition}\label{unanimity-splitting}
In a unanimity WVG with $q=w(N)$, if Banzhaf indices are used as payoffs of agents in a WVG, then it is beneficial for an agent to split up into several agents.
\end{proposition}
\begin{proof}
In a WVG with $q=w(N)$, the Banzhaf index of each player is $1/n$. Let player $i$ split up into $m+1$ players. In that case there is a total of $n+m$ players and the Banzhaf index of each player is $1/(n+m)$. In that case the total Banzhaf index of the split up players is $\frac{m+1}{n+m}$, and for $n>1$, $\frac{m+1}{n+m}> \frac{1}{n}$. An exactly similar analysis holds for Shapley-Shubik index.
\end{proof}
However there is the  same motivation for all agents to split up into smaller players which would return the agents to parity.
\subsection{General case}

We recall that a player is critical in a winning coalition if the player's exclusion makes the coalition losing. We will also say that a player is \emph{critical for} a losing coalition $C$ if the player's inclusion results in the coalition winning.

%\begin{figure}
 % \begin{center}
    
  %  \includegraphics[scale=0.2]{playersplittingdiagram3}
  %\end{center}
  %\caption{\label{playersplitting}\normalsize Splitting of player $i$ into $i'$ and $i''$ .}
  %\end{figure}

\begin{proposition}\label{split-bound-proposition}
Let WVG $v$ be $[q;w_1,\ldots,w_n]$. If $v$ transforms to $v'$ by the splitting of player $i$ into $i'$ and $i''$, then
${\beta_{i'}}(v')+ {\beta_{i''}}(v')\leq 2{\beta_i}(v)$. Moreover, this upper bound is asymptotically tight. 
\end{proposition}
\begin{proof}
We assume that a player $i$ splits up into $i'$ and $i''$ and that $w_{i'}\leq w_{i''}$. We consider a losing coalition $C$ for which $i$ is critical in $v$. Then $w(C)<q\leq w(C)+w_i=w(C)+w_{i'}+w_{i''}$. 
\begin{itemize}
\item If $q-w(C)\leq w_{i'}$, then $i'$ and $i''$ are critical for $C$ in $v'$. 
\item If $w_{i'}<q-w(C)\leq w_{i''}$, then $i'$ is critical for $C\cup\{i''\}$ and $i''$ is critical for $C$ in $v'$. 
\item If $q-w(C)>w_{i''}$, then $i'$ is critical for $C\cup\{i''\}$ and $i''$ is critical for $C\cup\{i'\}$ in $v'$. 
\end{itemize}    
Therefore we have $\eta_{i'}(v')+\eta_{i''}(v')=2{\eta_i}(v)$ in each case.

Now we consider a player $x$ in $v$ which is other than player~$i$. If $x$ is critical for a coalition $C$ in $v$ then $x$ is also critical for the
corresponding coalition $C'$ in $v'$ where we replace $\{i\}$ by $\{i',i''\}$. Hence $\eta_x(v) \leq \eta_x(v')$. Of course $x$ may also be critical for some coalitions in $v'$ which contain just one of $i'$ and $i''$, so the above inequality will not in general be an equality. Moreover, 
\begin{eqnarray*}
{\beta_{i'}}(v')+{\beta_{i''}}(v')&=&\frac{{2\eta_{i}}(v)}{{2\eta_{i}}(v)+\sum_{x\in N(v') \setminus \{i', i''\}}{\eta_x}(v')}\\
&\leq&\frac{{2\eta_{i}}(v)}{2\eta_{i}(v)+\sum_{x\in N(v) \setminus \{i\}}{\eta_x}(v)} \\
&\leq&\frac{{2\eta_{i}}(v)}{\eta_{i}(v)+\sum_{x\in N(v) \setminus \{i\}}{\eta_x}(v)}=2{\beta_i}(v) 
\end{eqnarray*}

We can prove that this coefficient of $2$ is best possible. We take a WVG $[n;2,1,\ldots,1]$ with $n+1$ players. We find that ${\eta_{1}}=n+{n \choose 2}$ and for all other $x$, $\eta_x=1+{n-1 \choose 2}$. Therefore 
$$\beta_1 = \frac{n+{n \choose 2}}{n+{n \choose 2}+n(1+{n-1 \choose 2})}
 = \frac{n+1}{(n-2)^2} \sim 1/n .$$ 
 In case player $1$ splits up into $1'$ and $1''$ with weights $1$ each, then for all players $j$, $\beta_j=\frac{1}{n+2}$. Thus for large $n$, $\beta_{1'}+\beta_{1''}=\frac{2}{n+2} \sim 2\beta_1$.
\end{proof}

Moreover, we show that splitting into two players can decrease the Banzhaf index payoff by as much as a factor of almost $\sqrt{\frac{\pi}{2n}}$:

\begin{example}
\emph{Disadvantageous splitting}.
We take a WVG $v$ on $n$ players  where $v=[3n/2;2n,1,\dots ,1]$. For the sake of simplicity, we assume that $n$ is even. It is easy to see that player $1$ is a dictator. Now we consider the case where $v$ changes into $v'$ with player $1$ splitting up into $1'$ and $1''$ with weight $n$ each. For player $1'$ to be critical for a losing coalition in $v'$, the coalition much exclude $1''$ and have from $n/2$ to $n-1$ players with weight $1$ or it must include $1''$ and have from $0$ to $(n/2-1)$ players with weight 1. So $\eta_{1'}(v')=\eta_{1''}(v')=\sum_{i=0}^{n}{n-1 \choose i}=2^{n-1}$. Moreover, for a smaller player $x$ with weight $1$ to be critical for a coalition in $v'$, the coalition must include only one of $1'$ or $1''$ and $(n-2)/2$ of the $n-2$ other smaller players. So, $\eta_{x}(v')=2{n-2\choose (n-2)/2}$. By using Stirling's formula, we can approximate $\eta_{x}(v')$ by $\sqrt{\frac{2}{\pi{(n-2)}}}2^{n-1}$. We see that:

\begin{eqnarray*}
\beta_{i'}(v')&=&\beta_{i''}(v')\\
&\approx& \frac{2^{n-1}}{2^{n-1}+2^{n-1}+(n-1)\sqrt{\frac{2}{\pi{(n-2)}}}2^{n-1}}\\
&=&\frac{1}{2+\frac{(n-1)}{\sqrt{n-2}}\sqrt{\frac{2}{\pi}}}\\
&\sim&\sqrt{\frac{\pi}{2n}}.
\end{eqnarray*}

%So $\beta_{i'}(v')=\beta_{i''}(v')\approx \frac{2^{n-1}}{2^{n-1}+2^{n-1}+(n-1)\sqrt{\frac{2}{\pi{(n-2)}}}2^{n-1}}= \frac{1}{2+\frac{(n-1)}{\sqrt{n-2}}\sqrt{\frac{2}{\pi}}}\sim \sqrt{\frac{\pi}{2n}}$.

\end{example}

We notice that the bounds on the effect of splitting on the Banzhaf index are quite similar to those in the Shapley-Shubik case.

%\newpage
\section{Complexity of finding a beneficial split}\label{AAMAS-COMPLEXITY-SECTION}

From a computational perspective, it is \#P-hard for a manipulator to find the ideal splitting to maximize his payoff. An easier question is to check whether a beneficial split exists or not. We define a Banzhaf version of the BENEFICIAL SPLIT problem defined in \cite{edithideas}. \\

\noindent
\textbf{Name}: BENEFICIAL-BANZHAF-SPLIT   \\
\textbf{Instance}: $(v,i)$ where $v$ is the WVG $v=[q;w_1,\ldots, w_n]$ and player $i\in \{1,\ldots, n\}$.\\
\textbf{Question}: Is there a way for player $i$ to split his weight $w_i$ between sub-players $i_1, \ldots, i_m$ so that, in the new game $v'$, $\sum_{j=1}^{m}\beta_{i_j}(v')> \beta_{i}(v)$?\\

\begin{proposition}\label{BENEFICIAL-BZ-SPLIT}
BENEFICIAL-BANZHAF-SPLIT is NP-hard, and remains NP-hard even if the player can only split into two players with equal weights.
\end{proposition}
\begin{proof}
We prove this by a reduction from an instance of the classical NP-hard PARTITION problem to BENEFICIAL-BZ-SPLIT.
\\
\\
\textbf{Name}: PARTITION \\
\textbf{Instance}: A set of $k$ integer weights $A=\{a_1, \ldots, a_k \}$.\\
\textbf{Question}: Is it possible to partition $A$, into two subsets $P_1\subseteq A$, $P_2\subseteq A$ so that $P_1\cap P_2=\emptyset$ and $P_1\cup P_2=A$ and $\sum_{a_i\in A_1}a_i=\sum_{a_i\in A_2}a_i$?\\

Given an instance of PARTITION $\{a_1, \ldots, a_k\}$, we can transform it to a WVG $v=[q;w_1,\ldots, w_n]$ with $n=k+1$ where $w_i=8{a_i}$ for $i=1$ to $n-1$, $w_{n}=2$ and $q=4\sum_{i=1}^{k}a_i +2$. After that, we want to see whether it can be beneficial for player $n$ with weight $2$ to split 
into two sub-players $n$ and $n+1$ each with weight $1$ to form a new WVG $v'=[q;w_1,\ldots, w_{n-1},1,1]$. Note that, since the weights are integral, it is certainly not beneficial to split up a weight of $2$ other than into $1$ and $1$.

If $A$ is a `no' instance of PARTITION, then we see that no subset of the weights $\{w_1,\ldots, w_{n-1}\}$ can sum to $4\sum_{i}a_i$. This implies that player $n$ is a dummy. We see that even if player $n$ splits into sub-players, the sub-players are also dummies. Therefore $(v,n)$ is a `no' instance of BENEFICIAL-BZ-SPLIT.

Now let us assume that $A$ is a `yes' instance of PARTITION. In that case, let the number of subsets of weights $\{w_1,\ldots w_{n-1}\}$ summing to $4\sum_{i}a_i$ be $x$. Then $\eta_{n}(v)=x$. For $i\leq n-1$, player $i$ can be critical in winning coalition with weight exactly $q$ or more than $q$. 
We note that exactly half of the $x$ subsets of $\{w_1,\ldots w_{n-1}\}$ summing to 
$4\sum_{i}a_i$ contain $w_i$. If player $i$ is critical in a coalition $C$ which is a subset 
of $\{w_1,\ldots w_{n-1}\}$ then $i$ is also critical in $C \cup \{w_n\}$.
Therefore for $i\leq n-1$, $\eta_i(v)=\frac{x}{2}+2y_i$ where $y_i$ is the number of subsets of $\{w_1,\ldots w_{n-1}\}$ in which $i$ is critical. We see that 
$$\beta_{n}(v)=\frac{x}{x+\frac{kx}{2}+2y} 
{\rm\ where\ } \sum_{i\leq n-1}y_i=y.$$

However, in the new game $v'$, $\eta_{n}(v')=\eta_{n+1}(v')=x$ and for $i\leq n-1$, $\eta_i(v')=\frac{x}{2}+4y_i$, since there are now $4$ coalitions, $C,C\cup \{w_n\},
C\cup\{w_{n+1}\},C\cup\{w_n,w_{n+1}\}$, corresponding to each $C$. 
So 
$$\beta_n(v')+\beta_{n+1}(v')=\frac{2x}{2x+kx/2+4y}>\beta_{n}(v),$$
 since $x>0$. 
Thus, a `yes' instance of PARTITION implies a `yes' instance of BENEFICIAL-BZ-SPLIT.
\end{proof}

In terms of minimizing chances of manipulation, we see that computational complexity acts as a barrier. This idea of using computational complexity to model bounded rationality is well explained by Papadimitriou and Yannakakis \cite{195445}. In the context of complexity of voting, it was a series of groundbreaking papers by Bartholdi, Orlin, Tovey, and Trick \cite{bartholdi1991, bartholdi1989a, bartholdi1989b, bartholdi1992} that showed how important computationally complexity consideration is in terms of ease of computing winners and difficulty of manipulation.

\subsection{Pseudopolynomial algorithm}

It is well known that, although computing Banzhaf indices of players in a WVG is NP-hard, there are polynomial time algorithms using dynamic programming~\cite{matsui00survey} or generating functions~\cite{bilbao_generating} to compute Banzhaf indices if the weights of players are polynomial in $n$. Let this pseudo-polynomial algorithm be $\mathsf{BanzhafIndex}(v,i)$ which takes a WVG $v$ and an index $i$ as input and returns $\beta_i(v)$, the Banzhaf index of player $i$ in $v$. We use a similar argument as in \cite{edithideas} to show that a polynomial algorithm exists to find a beneficial split if the weights of players are polynomial in $n$ and the player $i$ in question can split into up to a constant $k$ number of sub-players with integer weights. Algorithm \ref{BeneficialSplitInWVG} takes as input a WVG $v$ and player $i$ which can split into a maximum of $k$ number of players. The algorithm returns `no' if no beneficial split exists and returns the optimal split otherwise. Whenever player $i$ in WVG $v$ splits according to a split $s$, we denote the new game by $v_{i,s}$.

We see that the total number of splits for player $i$ is equal to $q(w_i,k)$ where $q(n,k)$ is the partition function which gives the number of partitions of $n$ with $k$ or fewer addends. It is clear that for a constant $k$, the number of splits of player $i$ is less than ${(w_i)}^k$ which is a polynomial in $n$. Since the computational complexity for each split is also a polynomial in $n$, therefore Algorithm~\ref{BeneficialSplitInWVG} is  polynomial in $n$ if the weights are polynomial in $n$.

\begin{algorithm}[thb]
  \caption{BeneficialSplitInWVG}
  \label{BeneficialSplitInWVG}
  \textbf{Input:} $(v,i)$ where $v=[q;w_1,\ldots,w_n]$ and $i$ is the player which wants to split into a maximum of $k$ sub-players.

\textbf{Output:} Returns NO if there is no beneficial split. Otherwise returns the optimal split $(w_{i_1},\ldots,w_{i_{k'}})$ 
where $k'\leq k$, and $\sum_{j=1}^{k'}w_{i_j}=w_i$.

   \begin{algorithmic}[1]
   \STATE $\mathsf{BeneficialSplitExists}= \FALSE$
   \STATE $\mathsf{BestSplit}=\emptyset$
   \STATE $\mathsf{BestSplitValue}=-\infty$
   \STATE $\beta_i= \mathsf{BanzhafIndex}(v,i)$
   \FOR{$j=2$ to $k$}
 		\FOR{Each possible split $s$ where $w_i=w_{i_1}+\ldots+w_{i_j}$}
 		\STATE $\mathsf{SplitValue}= \sum_{a=1}^{j}\mathsf{BanzhafIndex}{(v_{i,s},i_a)}$
  \IF{$\mathsf{SplitValue}>\beta_i$}
  \STATE $\mathsf{BeneficialSplitExists}=\TRUE$
  \IF{$\mathsf{SplitValue}>\mathsf{BestSplitValue}$}
  \STATE $\mathsf{BestSplit}=s$
  \STATE $\mathsf{BestSplitValue}=\mathsf{SplitValue}$
  \ENDIF 
   \ENDIF 
  
 	\ENDFOR
 	\ENDFOR 
  \IF{$\mathsf{BeneficialSplitExists}= \FALSE$}
   \RETURN $\FALSE$
   \ELSE
  \RETURN $\mathsf{BestSplit}$ 
  \ENDIF 
  \end{algorithmic}
\end{algorithm}

\section{Merging and annexation}\label{MERGING-SECTION-AAMAS}

For the case of players merging to gain advantage, we examine two cases. One is annexation where one voter takes the voting weight of other voters. The annexation is advantageous if the payoff of the new merged coalition in the new game is greater than the payoff of the annexer in the original game. The other case is voluntary merging where voters merge to become a bloc for which their new payoff is more than the sum of their individual payoffs. For every game $(N,v)$, the result of the merging of players in coalition $S$ is another game $((N\setminus S)\cup \{\&S\},v_{\&S})$.

We define the problem of checking a beneficial voluntary merge or annexation:\\

\noindent
\textbf{Name}: BENEFICIAL-BZ-MERGE\\
\textbf{Instance}: $(v,S)$ where $v$ is the WVG $v=[q;w_1,\ldots, w_n]$ and $S\subset N$. \\
\textbf{Question}: Suppose coalition $S$ merges to form a new game $((N\setminus S)\cup \{\&S\},v_{\&S})$. Is $\beta_{\&S}(v_{\&S})>\sum_{i\in S}\beta_i(v)$?\\

%For every game $(N,v)$, the result of merging of players in coalition $S$ is another game $((N\setminus S)\cup \{\&S\},v_{\&S})$.

\noindent
\textbf{Name}: BENEFICIAL-BZ-ANNEXATION\\
\textbf{Instance}: $(v,S,i)$ where $v$ is the WVG $v=[q;w_1,\ldots, w_n]$, $i$ is the $i$th player in $v$ and $S\subset (N\setminus \{i\})$. \\
\textbf{Question}: If $i$ annexes coalition $S$ to form a new game $((N\setminus (S\cup\{i\}))\cup \{\&(S\cup\{i\})\},v_{\&(S\cup\{i\})})$, is $\beta_i(v_{\&(S\cup\{i\})})>\beta_i(v)$?\\

If Shapley-Shubik indices are used as payoffs in place of Banzhaf indices, then  the corresponding problems are defined with BZ replaced by SS so that BENEFICIAL-SS-MERGE corresponds to BENEFICIAL-BZ-MERGE. 
Felsenthal and Machover \cite{mainbook} prove that if a player annexes other players, then it cannot be the case that the annexation is disadvantageous if the Shapley-Shubik indices are used as payoffs. We provide a clear and simple proof of this theorem. Let player $i$ be critical for a coalition $S$ in WVG $v$. Then the contribution to $\phi_i(v)$ from this is $\frac{(|S|-1)!(n-|S|)!}{n!}$. We consider a game $v_{\&\{i,j\}}$ where $i$ annexes $j$. For every $S$ for which $i$ is critical in $v$, the contribution to $\phi_{\&\{i,j\}}(v_{\&\{i,j\}})$ is either $\frac{(|S|-2)!(n-|S|)!}{(n-1)!}$ or $\frac{(|S|-1)!(n-|S|-1)!}{(n-1)!}$. For either case we see that $\phi_{\&\{i,j\}}(v_{\&\{i,j\}})>\phi_i(v)$. However Felsenthal and Machover \cite{mainbook} show that, for the case of the Banzhaf index, annexation could be disadvantageous. They provide a 13-player WVG for which annexation is disadvantageous, which is the simplest example they could find. We provide an 8-player WVG where annexation is disadvantageous:

\begin{example}
In WVG $[13;7,6,1,1,1,1,1,1]$, player $1$ has Banzhaf index $0.48507$. If player $1$ annexes one of the small players, the new game is $[13;8,6,1,1,1,1,1]$ and the Banzhaf index becomes $0.47826$.
\end{example}

For the case where the merging is voluntary instead of an annexation,  for both the Banzhaf index and Shapley-Shubik index, merging can be advantageous or disadvantageous. As in the case of splitting, we expect it to be hard to find a beneficial merge:

\begin{proposition}\label{BENEFICIAL-BZ-MERGE}
BENEFICIAL-BZ-MERGE is NP-hard.
\end{proposition}
\begin{proof}
Given an instance of PARTITION $\{a_1, \ldots, a_k\}$, we can transform it to a WVG $v=[q;w_1,\ldots, w_n]$ where $n=k+3$, $w_i=8{a_i}$ for $i=1$ to $n-3$, $w_{n-2}=w_{n-1}=w_n=1$, and $q=4\sum_{i=1}^{k}a_i+2$.  

If $A$ is a `no' instance of PARTITION, then we see that a subset of weights $\{w_1,\ldots w_{n-3}\}$ cannot sum to $4\sum_{i}a_i$. This implies that players $(n-2),(n-1)$ and $n$ are dummies. Even if players $n$ and $(n-1)$ merge together, the new player $\&\{n-1,n\}$ remains a dummy in the new game $v_{\&\{n-1,n\}}$. 

Now let us assume that $A$ is a `yes' instance of PARTITION. In that case, let the number of subsets of weights $\{w_1,\ldots w_{n-3}\}$ summing to $4\sum_{i}a_i$ be $x$. For $i\leq n-3$, player $i$ can be critical in winning coalitions with weight $q$ or $q+1$ 
or more than $q+1$. The number of coalitions for the first two cases are $3x/2$ and $x/2$, 
respectively, corresponding to the participation of either $2$ or $3$ of the unit players. 
The third case corresponds to coalitions in which the three unit players are dummies. 
Therefore for $i\leq n-3$, $\eta_i=\frac{4x}{2}+8y_i$ where $y_i$ is the number of subsets of $\{w_1,\ldots,w_{n-3}\}$ in which $i$ is critical. Moreover, $\eta_{n-2}(v)=\eta_{n-1}(v)=\eta_{n}(v)=2x$, since each unit player is critical only when exactly one other of these 
is in the coalition. 
Then 
$$\beta_{n}(v) = \frac{2x}{6x+\frac{4kx}{2}+8y}, 
{\rm\ where\ } \sum_{i\leq n-3}y_i=y.$$

In the new game $v_{\&\{n-1,n\}}$, $\eta_{\&\{n-1,n\}}(v_{\&\{n-1,n\}})$ is $2x$ but $\eta_{n-2}(v_{\&\{n-1,n\}})$ is $0$. For $i\leq n-3$, $\eta_{i}(v_{\&\{n-1,n\}})$ is $\frac{2x}{2}+4y_i$. We see that 
$$\beta_{\&\{n-1,n\}}(v_{\&\{n-1,n\}}) = \frac{2x}{2x+2kx/2+4y}.$$ 
Therefore, 
$$\beta_{\&\{n-1,n\}}(v_{\&\{n-1,n\}})> \beta_{n}(v)+\beta_{n-1}(v),$$ 
which means that $n$ and $(n-1)$ had a beneficial merge. It has been shown that a `yes' instance of PARTITION implies a `yes' instance of BENEFICIAL-BZ-MERGE.
\end{proof}

\begin{proposition}
BENEFICIAL-BZ-ANNEXATION is \allowbreak NP-hard.
\end{proposition}
\begin{proof}
Given an instance of PARTITION, $\{a_1, \ldots, a_k\}$, we can transform it to a WVG $v=[q;w_1,\ldots, w_n]$ where $n=k+2$, $w_i=8{a_i}$ for $i=1$ to $n-2$, $w_{n-1}=1$, $w_n=1$ and $q=4\sum_{i=1}^{k}a_i+2$. Just as in Proposition~\ref{BENEFICIAL-BZ-MERGE}, we see that a `no' instance of partition implies that $w_{n-1}$ and $w_{n}$ are dummies even if $n$ annexes $(n-1)$. However, a `yes' instance of partition implies that player $n$ benefits by annexing player $(n-1)$.
\end{proof}

\begin{proposition}
BENEFICIAL-SS-MERGE is NP-hard
\end{proposition}

\begin{proof}
Given an instance of PARTITION $\{a_1, \ldots, a_k\}$, we can transform it to a WVG $v=[q;w_1,\ldots, w_n]$ where $n=k+3$, $w_i=8{a_i}$ for $i=1$ to $n-2$, $w_{n-2}=w_{n-1}=w_n=1$, and $q=4\sum_{i=1}^{k}a_i+2$.

If $A$ is a `no' instance of PARTITION, then we see that a subset of weights $\{w_1,\ldots w_{n-3}\}$ cannot sum to $4\sum_{i}a_i$. This implies that players $(n-2)$, $(n-1)$ and $n$ are dummies. Even if player $n$ and $(n-1)$ merge together, the new player $\&\{n-1,n\}$ remains a dummy in the new game $v_{\&\{n-1,n\}}$.

Now let us assume that $A$ is a `yes' instance of PARTITION. For each partition $(P_1,P_2)$ where $|P_1|=p_1$ and $|P_2|=p_1$, we check the number of permutations corresponding to $(P_1,P_2)$. In the original game $v$, the contribution to the Shapley-Shubik payoff for either player $n$ or $(n-1)$ by the permutations corresponding to $(P_1,P_2)$ is 
$$\frac{2(p_1+1)!(p_2+1)!}{n!}=\frac{p_1!p_2!}{n!}2(p_1+1)(p_2+1).$$ 
If players $n$ and $n-1$ merge into bloc $\&\{n-1,n\}$, then the contribution to the Shapley-Shubik payoff to bloc $\&\{n-1,n\}$ by the permutations corresponding to $(P_1,P_2)$ is 
$$\frac{p_1!(p_2+1)!+(p_1+1)!p_2!}{(n-1)!}=\frac{p_1!p_2!}{n!}(n(p_1+1+p_2+1)).$$ 
For the merge to be beneficial, it is required that the sum of the Shapley-Shubik indices 
of $(n-1)$ and $n$ in the original game $v$ is less than the Shapley-Shubik index of 
$\&\{n-1,n\}$ in the game $v_{\&\{n-1,n\}}$, i.e., 
$4(p_1+1)(p_2+1) < n(p_1+1+p_2+1)$.  Since $(p_1+1)+(p_2+1)=n-1$, we have 
$$4(p_1+1)(p_2+1)\leq 4\left({\frac{n-1}{2}}\right)^2 < n(n-1) = n(p_1+1+p_2+1),$$ 
and so $\phi_{n-1}(v)+\phi_{n}(v) < \phi_{\&\{n-1,n\}}(v_{\&\{n-1,n\}})$.
\end{proof}

We examine the limits of advantage or disadvantage for the case of the annexation of another player to increase the Banzhaf index.

\begin{proposition}\label{annexation-banzhaf-bounds}
$ \frac{\beta_i(v)+\beta_j(v)}{2} \leq \beta_i(v_{\&(\{i,j\})})\leq 1$.
\end{proposition}
\begin{proof}

Let $v$ be WVG $[q;w_1,\ldots, w_n]$. Suppose $i$ annexes or merges 
with player $j$ and $v'$ is $v_{\&(\{i,j\})}$. Then the new game is 
$((N\setminus \{j\})\cup \{\&(\{i,j\})\},v')$. 
From the proof of Proposition~\ref{split-bound-proposition}, 
we see that $\eta_{\&(\{i,j\})}(v')$ equals $\frac{1}{2}(\beta_i(v)+\beta_j(v))$. 

Now consider a player $x$ which is other than player $i$ or player $j$. Let $S$ be coalition such that $S\subseteq N\setminus\{i,j,x\}$. If $x$ is critical for $S$ in $v$ then $x$ is critical for $S$ in $v'$. If $x$ is critical for $S\cup \{i,j\}$ in $v$ then $x$ is critical for $S\cup \&(\{i,j\})$ in $v'$. However, $x$ may also be critical for $S\cup \{i\}$ or $S\cup \{j\}$ in $v$. 
So $\eta_x(v)\geq \eta_x(v')$. We see that:

\begin{eqnarray*}
\beta_{\&(\{i,j\})}(v')&=&\frac{\eta_{\&(\{i,j\})}(v')}{\eta_{\&(\{i,j\})}(v')+\sum_{x\in  (N\setminus\{i,j\})}\eta_x(v')}\\
&=&\frac{\frac12(\eta_i(v)+\eta_j(v))}{\frac12(\eta_i(v)+\eta_j(v))+\sum_{x\in  (N\setminus\{i,j\})}\eta_x(v')}\\
&\geq&\frac{\frac{1}{2}(\eta_i(v)+\eta_j(v))}{\eta_i(v)+\eta_j(v)+\sum_{x\in  (N\setminus\{i,j\})}\eta_x(v)}\\
&=&\frac{1}{2}(\beta_{i}(v)+\beta_j(v)).
\end{eqnarray*}

The upper bound is tight and easy to observe. If player $i$ is a dummy and $j$ is a dictator then $\beta_{i}(v)=0$ whereas $\beta_i(v')=1$. The upper bound can also be achieved by two big enough players joining forces.
\end{proof}

We have seen that annexation can be disadvantageous in the case of the Banzhaf index. One would at least expect that the Banzhaf index payoff after annexing another player to be monotone in the power of the annexed player. Surprisingly, this is not the case. Suppose $w_i \geq w_j \geq w_k$ in a WVG $v$. We provide an example where $\beta_{i,k}>\beta_{i,j}$. 
We call this the \emph{annexation non-monotonicity paradox}:

\begin{example}
In the WVG $[9;3,3,2,1,1,1]$ we see that player~$2$ has more weight 
than player~$3$. However if player $1$ annexes player $2$ to form game 
$[9;6,2,1,1,1]$, its Banzhaf index is $0.4$, whereas if player $1$ annexes 
player $3$ to form game $[9;5,3,1,1,1]$, its Banzhaf index is $0.411765$.
\end{example}

\begin{proposition}
For any coalition, $S\subset N\setminus \{i\}$, $ \phi_i(v) \leq \phi_i(v_{\&(\{i\}\cup S)})\leq 1$.
\end{proposition}
\begin{proof}
The lower bound follows from the result by Felsenthal and Machover \cite{mainbook} that annexation cannot decrease the Shapley-Shubik index of a player. Moreover, the upper bound is tight and easily attainable if $\{i\}\cup S$ is big enough.
\end{proof}

\begin{proposition}
For the unanimity game and for both the Shapley-Shubik index and Banzhaf index:
\begin{enumerate}
\item it is disadvantageous for a coalition to merge;
\item it is advantageous for a player to annex.
\end{enumerate}
\end{proposition}
\begin{proof}
We check each case separately:
\begin{enumerate}
\item This is expected considering Proposition~\ref{unanimity-splitting}.  If $k$ players merge, then the payoff of the new coalition is $1/(n-k+1)$. It is easy to see that $1/(n-k+1)<k/n$.
\item For a unanimity WVG with $n$ players, the payoff of each player is $1/n$. If a player annexes $k-1$ other players, its payoff is $1/(n-k+1)$ which is more than $1/n$.
\end{enumerate}
\end{proof}

\begin{table*}[t]
\centering
\caption{Complexity of False Name Manipulations in WVGs}
\begin{tabular}{|l|c|c|} \hline
%\small
&Banzhaf index&Shapley-Shubik Index\\ \hline
SPLITTING&NP-hard&NP-hard \cite{edithideas}\\ \hline
MERGING&NP-hard&NP-hard\\ \hline
ANNEXATION&NP-hard& advantageous \cite{mainbook}\\ \hline
SPLITTING in unanimity game&advantageous&advantageous \cite{edithideas}\\ \hline
MERGING in unanimity game&disadvantageous&disadvantageous\\ \hline
ANNEXATION in unanimity game&advantageous& advantageous\\ \hline
\end{tabular}
\label{ComplexityofFalseNameManipulations}
\end{table*}

In a WVG, if player $i$ annexes a dummy, then there is no difference to the Banzhaf index payoff of each player. This is because the Banzhaf value of each player reduces to half of the original Banzhaf value. Moreover, it follows from Proposition~\ref{annexation-banzhaf-bounds} that if a player annexes a player bigger than itself, its Banzhaf index can only increase. Thus annexation could only be disadvantageous, if a player annexes a smaller player. Although, deciding a beneficial merge or annexation is computationally difficult, it may often be easier in practice. We propose a simple heuristic to get beneficial annexations or at least to avoid disadvantageous annexations. It appears to be a better strategy to annex fewer players with some total weight than more players with the same total weight. This is because, while annexing, the annexer does not want to increase the payoff of other players significantly. 

\section{Conclusions}

Weighted voting games are important game models in multiagent systems. We have investigated the impact on the Banzhaf power distribution due to a player splitting into smaller players in a weighted voting game. We have also considered the case of manipulation via annexation and voluntary merging when the payoff is according to the Banzhaf index or the Shapley-Shubik index. Both the complexity of manipulation and the limits of manipulation are examined. The complexity results are summarised in 
Table~\ref{ComplexityofFalseNameManipulations}. The Shapley-Shubik index appears to be a more desirable solution because annexation does not decrease the payoff of a player. It is seen that manipulation may be discouraged by keeping weights which are large or non-integers. The finer, more detailed, analysis for players splitting into more than two players or merging into bigger blocs is still unexplored. Although, it is NP-hard to evaluate different false-name manipulations, it may be the case that certain instances of WVGs are more susceptible to manipulation \cite{Aziz-classification}. A careful investigation of heuristics for false-name manipulation is also a promising area of research. There is scope to analyse such false-name manipulations with respect to other cooperative game-theoretic solutions. A particularly suitable solution to consider could be the nucleolus which not only always exists but is also unique. Further examination into various aspects of manipulation in weighted voting games promises to give better insight into designing fairer and manipulation-resistant systems. Another interesting question is to what extent can the results be applied to more general cooperative games.


\section*{Acknowledgment}

Partial support for this research was provided by DIMAP (the Centre for Discrete Mathematics and its Applications). DIMAP is funded by the UK EPSRC 
under grant EP/D063191/1. The first author would also like to thank the Pakistan National ICT R \& D Fund for funding his research. We are also thankful to John Fearnley and the anonymous referees for valuable comments on the earlier draft of this paper.

\bibliographystyle{plain}

{\bf HARIS AZIZ}
\\
{\it Department of Computer Science, University of Warwick, %\\
Coventry CV4 7AL, United Kingdom.\\
haris.aziz@warwick.ac.uk.}

{\bf MIKE PATERSON}
\\
{\it Department of Computer Science, University of Warwick, %\\
Coventry CV4 7AL, United Kingdom.\\
msp@dcs.warwick.ac.uk.}

\end{document}